\documentclass[12pt]{article}
\usepackage{amssymb,amsmath,color,amsmath,bm}
\usepackage{enumitem}
\usepackage[colorlinks,linkcolor=blue]{hyperref}
\usepackage{geometry}
\usepackage{tabularx}
\usepackage{bm}
\usepackage{authblk}
\usepackage{nicematrix}
\usepackage{tikz}
\usepackage[square, comma, sort&compress, numbers]{natbib}
\usepackage{nicematrix}
\usepackage{float} 
\usepackage{multirow} 
\usepackage{booktabs} 
\newtheorem{theorem}{Theorem}[section]
\newtheorem{definition}{Definition}[section]
\newtheorem{lemma}{Lemma}[section]
\newtheorem{example}{Example}[section]

\newtheorem{remark}{Remark}[section]

\catcode`@=11 \@addtoreset{equation}{section} \catcode`@=12
\begin{document}
	\title{\bf New Construction of Power Functions with Low $c$-Differential Uniformity over Finite Fields}
	\author{Zhiye Yang$^1$, Yan Wang$^2$,  Keqin Feng$^3$\\
		~\\
		$^1$ Research Center for Number Theory and Its Applications\\
		School of Mathematics, Northwest University\\ Xi'an 710127,  Shaanxi, China\\
		E-mail address: zyyang02@126.com \\
		$^2$ School of Science, Xi’an University of Architecture and Technology,\\ Xi’an, 710311, Shaanxi, China\\ 
		E-mail address: wangyan@xauat.edu.cn\\
		$^3$ Department of Mathematical Sciences,Tsinghua University,\\ Beijing, 100084, China.\\
		E-mail address: fengkq@tsinghua.edu.cn\\} 
	\date{}
	\maketitle
	{\bf Abstract.}
	{\small 
This paper investigates the $c$-differential uniformity of power functions over finite fields, an important class of cryptographic functions with favorable differential properties. 
Specifically, for finite fields $\mathbb{F}_q$ satisfying $q-1=en$ with $e\ge 3$ and $e\mid n$, we prove that there exists 
 $c\in\mathbb{F}_q\setminus\{0,1,\epsilon,\dots,\epsilon^{e-1}\}$, where $\epsilon$ is an $e$-th primitive root of unity in $\mathbb{F}_q^*$, the constructed power functions $f(x)=x^{ln+1}$ with $1\le l\le e-1$ and $\gcd(l,e)=1$ satisfy the upper bound $\Delta(f,c)\le e$, provided that certain cyclotomic conditions hold.  
 We show that our conditions are mild; 
 namely, such power functions can be constructed over infinitely many extension fields $\mathbb{F}_q$ of $\mathbb{F}_p$ for any given $e\ge 3$ and prime $p$ with $p\nmid e$. 
  Furthermore, based on the Weil bound for multiplicative character sums, we prove that the obtained upper bound is tight for sufficiently large $q$, demonstrating the optimality of our results. We also analyze the special case $c=-1$ and derive simplified explicit conditions. In particular, we explicitly characterize the admissible parameters for the case $e=3$ and present concrete function examples for practical validation. 


\section{Introduction}
 
 

 Differential cryptanalysis, introduced by Biham and Shamir\cite{biham1991}, is one of the most powerful statistical attacks against symmetric‑key block ciphers. To evaluate the resistance of cryptographic functions against such attacks, Nyberg\cite{nyberg1994} introduced the concept of \emph{differential uniformity} in 1993. For a function $f\colon \mathbb{F}_q\to \mathbb{F}_q$, the differential uniformity is defined as the maximum, over all $a\in\mathbb{F}_q^*$ and $b\in\mathbb{F}_q$, of the number of solutions $x\in\mathbb{F}_q$ to the equation $f(x+a)-f(x)=b$. 
 
 Recently, Ellingsen et al.\cite{elling2020tit} introduced a generalization of the classical differential, known as the multiplicative $c$-differential, which considers equations of the form $f(x+a)-cf(x)=b$. This new framework is motivated by cryptographic primitives that employ multiplication as a fundamental operation, such as the IDEA cipher\cite{borisov2002}. The \emph{$c$-differential uniformity} is formally defined as follows.
 \begin{definition}\label{defcdiff}
 	Let $q=p^m$, where $p$ is a prime and $m\ge 1$.
 	Let $\mathbb{F}_q$ denote the finite field with $q$ elements, and let $f\colon \mathbb{F}_q\to\mathbb{F}_q$ be a function.
 	For $a,b,c\in\mathbb{F}_q$ with $c\neq 0$, we denote by $N(f,c,b,a)$ the number of solutions $x\in\mathbb{F}_q$ of the equation
 	\begin{equation}\label{eq:c_diff_eq}
 		f(x+a)-c f(x)=b.
 	\end{equation}
 	
 	For each $c\in\mathbb{F}_q^*$, the $c$-differential uniformity $\Delta(f,c)$ of $f$ over $\mathbb{F}_q$ is defined as
 	\[
 	\Delta(f,c)=\max\bigl\{ N(f,c,b,a)\,\big|\, a,b\in\mathbb{F}_q,\text{ and } a\neq 0 \text{ if } c=1\bigr\}.
 	\]
 	
 	If $\Delta(f,c)=1$, then $f$ is said to be a perfect $c$-nonlinear (PcN) function over $\mathbb{F}_q$.
 	If $\Delta(f,c)=2$, then $f$ is said to be an almost perfect $c$-nonlinear (APcN) function over $\mathbb{F}_q$.
 \end{definition}
  A smaller $c$-differential uniformity $\Delta(f,c)$ corresponds to better resistance to differential attacks. Accordingly, functions with low $c$-differential uniformity have received significant attention.
  
  
 A polynomial $f\in\mathbb{F}_q[x]$ is a \emph{permutation polynomial} over $\mathbb{F}_q$ if the induced map $f\colon \mathbb{F}_q\to\mathbb{F}_q$ is bijective. In addition, $f$ is said to be an \emph{involution} if it is its own compositional inverse, i.e., $f(f(x)) = x$ for all $x\in\mathbb{F}_q$\cite{lidl1997FFA}.
 Permutation polynomials have wide applications in cryptography, coding theory, and combinatorial designs. In particular, those with low 
 $c$-differential uniformity serve as excellent candidates for S-boxes in block ciphers\cite{carlet2020boolean}. Among permutation polynomials over finite fields, power functions $f(x)=x^d$ constitute one of the most widely investigated families, due to their simple algebraic structure and low implementation cost in cryptographic applications. Recall that $f(x)=x^d$ is a permutation of $\mathbb{F}_q$ if and only if $\gcd(d,q-1)=1$. 
 
 Over the past decade, substantial research has been devoted to the study of 
$c$-differential uniformity $\Delta(f,c)$ for various classes of functions; we refer the reader to the survey \cite{Mesnager2022CCDS} for a comprehensive overview. In particular, power functions with low $c$-differential uniformity have attracted considerable interest due to their efficiency in hardware implementations. For the convenience of the reader, we summarize some known results on the $\Delta(f,c)$ of power functions in Table \ref{fsummary}, where we also include our new results for comparison. 
For further results, we refer the reader to the references cited in the table, as well as the representative works~\cite{YanZhang2022DCC,YanM2023tit,TuL2023tit,Yan2024Dm}. 

\begin{table}[htbp!]
	\centering
	\caption{Known and new results on the $\Delta(f,c)$ of $f(x)=x^d$ over $\mathbb{F}_{p^m}$ with $c\neq 1$}
	\label{fsummary}
	\setlength{\tabcolsep}{3.8pt}
	\footnotesize
	\begin{tabular}{|c|c|c|c|c|}
		\hline
		$p$ & $d$ & conditions & $\Delta(f,c)$ & Refs. \\
		\hline
		$>2$ & $2$ & none & $2$ & \cite{elling2020tit} \\
		\hline
		any & $p^m-2$ & $c=0$ & $1$ & \cite{elling2020tit} \\
		\hline
		$2$ & $2^m-2$ & $c\neq 0,\ \mathrm{Tr}_1^m(c)=\mathrm{Tr}_1^m(c^{-1})=1$ & $2$ & \cite{elling2020tit} \\
		\hline
		odd & $p^m-2$ & $\begin{aligned}c\neq0,&\ (c^2-4c)\notin [\mathbb{F}_{p^m}]^2,\,\\
			&(1-4c)\notin [\mathbb{F}_{p^m}]^2,\,\text{or}\,\,c\neq4,4^{-1}\end{aligned}$ & $2$ & \cite{elling2020tit} \\
		\hline
		$3$ & $\dfrac{3^k+1}{2}$ & $c=-1,\ \dfrac{2m}{\gcd(k,2m)}=1$ & $1$ & \cite{elling2020tit} \\
		\hline
		$2$ & $2^m-2$ & $c\neq 0,\ \mathrm{Tr}_1^m(c)=0,\text{or }\mathrm{Tr}_1^m(c^{-1})=0$ & $3$ & \cite{elling2020tit} \\
		\hline
		odd & $p^m-2$ & $\begin{aligned}c\neq0,4,&\,4^{-1}\ (c^2-4c)\in [\mathbb{F}_{p^m}]^2\,\\
			&\text{or}\,(1-4c)\in [\mathbb{F}_{p^m}]^2\end{aligned}$ & $3$ & \cite{elling2020tit} \\
		\hline
		odd & $\dfrac{p^2+1}{2}$ & $c=-1,\,m$ odd & $1$ & \cite{Bartoli2020JAC} \\
		\hline
		odd & $p^2-p+1$ & $c=-1,\,m=3$ & $1$ & \cite{Bartoli2020JAC} \\
		\hline
		odd & $\begin{aligned}p^4&+(p-2)p^2\\
			&+(p-1)p+1\end{aligned}$ & $c=-1,\,m=5$ & $1$ & \cite{Hasan2021DCC}\\
		\hline
		odd & $(p^5+1)/(p+1)$ & $c=-1,\,m=5$ & $1$ & \cite{Hasan2021DCC} \\
		\hline
		odd & $\begin{aligned}&(p-1)p^6+p^5+(p-2)p^3\\
			&+(p-1)p^2+p;\,(p^7+1)/(p+1)\end{aligned}$ & $c=-1,\,m=7$ & $1$ & \cite{Hasan2021DCC} \\
		\hline
		odd & $(p^7+1)/(p+1)$ & $c=-1,\,m=7$ & $1$ & \cite{Hasan2021DCC} \\
		\hline
		odd & $\dfrac{p^l+1}{2}$ & $\begin{aligned}&c=-1,\,(l,2m)=1, \\
			&p\equiv 1\pmod{4},\text{or\,}\,p\equiv 3\pmod{8}\end{aligned}$ & $\frac{p+1}{2}$ & \cite{Hasan2021DCC} \\
		\hline
		$3$ & $x^{\frac{p^m+7}{2}}$ & $c=-1,\,m$ odd & $\le 2$  & \cite{LN2021DCC} \\
		\hline
		any & $d^{-1}\pmod{p^m-1}$ & $x^d$ is PcN, $c'=c^d$ & $1$ & \cite{Wang2022DAM} \\
		\hline
		$2$ & $\{2^j\}_{j\ge0},\,\{2^j(2^k+1)\}_{k,j\ge0}$ & $c\neq 1$, resp., $c\in\mathbb{F}_{2^{\gcd(k,m)}}\setminus\{1\},\,v_2(m)\le v_2(k)$ & $1$  & \cite{Wang2022DAM} \\
		\hline
		$>2$ & $\text{odd}\,2(p^k+1)^{-1}\pmod{p^m-1}$ & $c=-1$ & $1$  & \cite{Wang2022DAM} \\
		\hline
		$>2$ & $\dfrac{p^m+1}{2}\left(\dfrac{p^k+1}{2}\right)^{-1}$ & $c=-1,\ v_2(k)=v_2(m),\ p^m\equiv 1\pmod{4}$ & $1$  & \cite{Wang2022DAM} \\
		\hline
		$3$ & $\begin{aligned}\dfrac{3^k+1}{2}d&\equiv \dfrac{3^m+1}{2}\pmod{3^m-1}\\
			d&\ \text{odd}\end{aligned}$ & $k\text{ and }m\text{ are odd}
		,\,\gcd(m,k)=1$ & $1$ & \cite{Zha2021DCC} \\
		\hline
		$5$ & $\begin{aligned}\dfrac{5^k+1}{2}d&\equiv \dfrac{5^m+1}{2}\pmod{5^m-1},\\\,
			&d\ \text{odd}\end{aligned}$ & $\begin{aligned}k&\text{ and }m\text{ are positive}\\
			&\text{integer such that } \gcd(2m,k)=1\end{aligned}$ & $1$ & \cite{Zha2021DCC} \\
		\hline
		odd & $\dfrac{p^m+1}{2}(p^k+1)^{-1}$ & $\begin{aligned}&d \text{ even },\,c=-1 ,p^m\equiv 3\pmod{4}\\&\text{or}\,\,d \text{ odd },\,c=-1 ,p^m\equiv 3\pmod{4}
		\end{aligned}$ & $\begin{aligned}&\leq6\\\text{or}\,\,&\leq3
		\end{aligned}$ & \cite{Zha2021DCC} \\
		\hline
		odd & $\dfrac{p^m+1}{4}+\dfrac{p^m-1}{2}$ & $\,c=-1 ,p^m\equiv 7\pmod{8}$ & $\leq3$ & \cite{Zha2021DCC} \\
		\hline
		odd & $\dfrac{p^m-1}{2}+p^k+1$ & $\begin{aligned}\,c=-1 ,\dfrac{m}{(m,k)}\,&\text{odd},\,p^m\equiv 3\pmod{4}\\
			&\text{or }\,\,p^m\equiv 1\pmod{4}\end{aligned}$  & $\begin{aligned}&\leq3\\\text{or}\,\,&\leq6
		\end{aligned}$ & \cite{Zha2021DCC} \\
		\hline
		$3$ & $\dfrac{3^{\frac{m+1}{2}}-1}{2}\,\text{ and }\, \dfrac{3^{\frac{m+1}{2}}-1}{8}$ & $c=-1,\,m\equiv 1\pmod{4}$ & $\le 2$  & \cite{Yan2022ccds} \\
		\hline
		$3$ & $\dfrac{3^{\frac{m+1}{2}}-1}{2}+\dfrac{3^m-1}{2}$ & $c=-1,\,m\equiv 3\pmod{4}$ & $\le 2$  & \cite{Yan2022ccds} \\
		\hline
	\end{tabular}
\end{table}
\begin{table}[htbp!]
	\centering
	\setlength{\tabcolsep}{3.8pt}
	\footnotesize
	\begin{tabular}{|c|c|c|c|c|}
		\hline
		$p$ & $d$ & conditions & $\Delta(f,c)$ & Refs. \\
		\hline
		$3$ & $\dfrac{3^{\frac{m+1}{2}}-1}{8}+\dfrac{3^m-1}{2}$ & $c=-1,\,m\equiv 3\pmod{4}$ & $\le 2$  & \cite{Yan2022ccds} \\
		\hline
		$3$ & $\big(3^{\frac{m+1}{4}}-1\big)\big(3^{\frac{m+1}{2}}+1\big)$ & $c=-1,\,m\equiv 3\pmod{4}$ & $\le 4$ & \cite{Yan2022ccds} \\
		\hline
		$3$ & $\dfrac{3^{m+1}+1}{4}+\dfrac{3^m-1}{2}$ & $c=-1,\,m$ odd & $\le 4$ & \cite{Yan2022ccds} \\
		\hline
		$3$ & $\dfrac{3^m+3}{2}$ & $c=-1,\,m$ even & $2$ & \cite{Mesnager2021tit} \\
		\hline
			$odd$ & $\frac{p^m+1}{2}$ & $\begin{aligned}&c\neq \pm1,\,p^m\equiv 3\pmod{4},\\
			\,\,&\eta(\frac{1-c}{1+c})=1\end{aligned}$ & $\leq 2$ & \cite{Mesnager2021tit} \\
		\hline
		$2$ & $2^k+1$ & $\begin{aligned}&c\in\mathbb{F}_{2^{(m,k)}}\backslash\{1\},\frac{m}{(m,k)}\geq3,\,(m\,\text{odd}),\\\,&\frac{m}{(m,k)}\geq4,\,(m\,\text{even})\,\frac{2^{(m,2k)}-1}{2^{(m,k)}-1}\end{aligned}$ & $2^{(m,k)}+1$ & \cite{Mesnager2021tit} \\
		\hline
		any & $p^k+1$ & $1\neq c\in\mathbb{F}_{p^{(m,k)}}$ & $p^{(m,k)}+1$ & \cite{Mesnager2021tit} \\
		\hline
		odd & $\frac{p^k+1}{2}$ & $c=-1$ & $p^{(m,k)}+1$ & \cite{Mesnager2021tit} \\
		\hline
		odd & $\frac{p^m+1}{2}$ & $c\neq\pm 1$ & $\leq 4$ & \cite{Mesnager2021tit} \\
		\hline
		any & $\frac{2p^m-1}{3}$ & $p^m\equiv 2\pmod{3},\,\,c\neq 1$ & $\leq 3$ & \cite{Mesnager2021tit} \\
		\hline
		$>3$ & $\frac{p^m+3}{2}$ & $\begin{aligned}&c=-1,\,p^m\equiv 3\pmod{4}\\
			\text{or}\,\,&c=-1,\,p^m\equiv 1\pmod{4}\end{aligned}$ & $\begin{aligned}&\leq3\\
			\text{or}\,\,&\leq4\end{aligned}$ & \cite{Mesnager2021tit} \\
		\hline
		odd & $\frac{p^m-3}{2}$ & $c=-1$ & $\leq 4$ & \cite{Mesnager2021tit} \\
		\hline
		3 & ${3^m-3}$ & $\begin{aligned}&c=-1,\,\,m\equiv 0\pmod{4}\\
			\text{or}\,\,&c=-1,\,m\not\equiv 0\pmod{4}\end{aligned}$ & $\begin{aligned}&6\\
			\text{or}\,\,&4\end{aligned}$ & \cite{Mesnager2021tit} \\
		\hline
		odd & $p^m-3$ & $c=-1$ & $\leq 6$ & \cite{WangY2026DAM} \\
		\hline
		any & $\begin{aligned}&\dfrac{l(p^m-1)}{e}+1,\,\,\\&(l,e)=1,\,\,1\leq l\leq e-1\end{aligned}$ & $\begin{aligned}&\{c\in \mathbb{F}_q\setminus\{0,1,\epsilon,\dots,\epsilon^{e-1}\}|1-c\epsilon^\lambda\in D_i;1-\epsilon^\lambda\in D_j\},\\
			&q-1=en,\,e\mid n,\,e\geq 3,\, \mathbb{F}_q^*=\langle\pi\rangle=\bigcup_{\lambda=0}^{e-1}D_\lambda,\,\epsilon=\pi^n \end{aligned}$ & $\leq e$  & \textcolor{red}{Theorem \ref{generalcons}} \\
		\hline
		any & $\begin{aligned}&\dfrac{l(p^m-1)}{e}+1,\,\,\\&(l,e)=1,\,\,1\leq l\leq e-1\end{aligned}$ & $\begin{aligned}&\{c\in \mathbb{F}_q\setminus\{0,1,\epsilon,\dots,\epsilon^{e-1}\}|1-c\epsilon^\lambda\in D_i;1-\epsilon^\lambda\in D_j\},\\
			&q-(e^{2e}-1)(2e(e-1)-1)\sqrt{q}\geq1 \end{aligned}$ & $= e$  & \textcolor{red}{Theorem \ref{th1tight}} \\
		\hline
		odd & $\begin{aligned}&\dfrac{l(p^m-1)}{e}+1,\,\,\\&(l,e)=1,\,\,1\leq l\leq e-1\end{aligned}$ & $c=-1,2\nmid e\geq3,\,2\in D_0,\,1-\epsilon^\lambda\in D_j,\,1\leq\lambda\leq e-1,$ & $\leq e$  & \textcolor{red}{Theorem \ref{th1case-1}} \\
		\hline
		odd & $\dfrac{p^m+2}{3}$ & $\begin{aligned}&\{c\in \mathbb{F}_q\setminus\{0,1,\epsilon,\epsilon^{2}\}|1-c\epsilon^\lambda\in D_i;\,\lambda=0,1,2\},\\
			&q-1=3n,\,6\mid n \end{aligned}$ & $\leq 3$  & \textcolor{red}{Theorem \ref{th1e=3}} \\
		\hline
		odd & $\dfrac{2p^m+1}{3}$ & $\begin{aligned}&\{c\in \mathbb{F}_q\setminus\{0,1,\epsilon,\epsilon^{2}\}|1-c\epsilon^\lambda\in D_i;\,\lambda=0,1,2\},\\
			&q-1=3n,\,6\mid n \end{aligned}$ & $\leq 3$  & \textcolor{red}{Theorem \ref{th1e=3}} \\
		\hline
		odd & $\dfrac{p^m+2}{3}$ & $c=-1,\,q=s^2+27t^2$ & $\leq 3$  & \textcolor{red}{Theorem \ref{th1e=3c=-1}} \\
		\hline
		odd & $\dfrac{2p^m+1}{3}$ & $c=-1,\,q=s^2+27t^2$ & $\leq 3$  & \textcolor{red}{Theorem \ref{th1e=3c=-1}} \\
		\hline
		2 & $\dfrac{p^m+2}{3}$ & $\begin{aligned}&\{c\in \mathbb{F}_q\setminus\mathbb{F}_4|\frac{1+c\epsilon}{1+c}\in D_0;\frac{1+c\epsilon^2}{1+c}\in D_0\,\},\\
			&q-1=3n,\,3\mid n,\,q=2^{6m} \end{aligned}$ & $\leq 3$  & \textcolor{red}{Theorem \ref{th1e=3p=2}} \\
		\hline
		2 & $\dfrac{2p^m+1}{3}$ & $\begin{aligned}&\{c\in \mathbb{F}_q\setminus\mathbb{F}_4|\frac{1+c\epsilon}{1+c}\in D_0;\frac{1+c\epsilon^2}{1+c}\in D_0\,\},\\
			&q-1=3n,\,3\mid n,\,q=2^{6m} \end{aligned}$ & $\leq 3$  & \textcolor{red}{Theorem \ref{th1e=3p=2}} \\
		\hline
	\end{tabular}
\end{table}
Most known constructions of power functions with low $c$-differential uniformity focus on the special case $c=-1$, and the cases $e=1$ 
 and $e=2$. 
 In contrast, our constructions work for general $e\ge 3$ and a large family of admissible $c$, and we prove the tightness of the bound. Our main result, Theorem \ref{generalcons}, gives a sufficient condition formulated in terms of cyclotomic cosets, under which the power function $f(x)=x^{ln+1}$ over $\mathbb{F}_q$ satisfies $\Delta(f,c)\le e$, where $q-1=en$, $e\mid n$, and $\gcd(l,e)=1$. We show that these conditions are mild: for any fixed integer $e\ge 3$ and prime $p\nmid e$, there exist infinitely many extension fields $\mathbb{F}_q$ of $\mathbb{F}_p$ fulfilling these requirements, see Theorem \ref{generalnot}. Moreover, employing the Weil bound for multiplicative character sums, we prove that this upper bound is tight for sufficiently large $q$, as stated in Theorem \ref{th1tight}. We further investigate the important special case $c=-1$, and for $e=3$, we provide a complete characterization: the exact number of admissible parameters $c$ satisfying the required condition is determined (Theorem \ref{cmany}), and concrete examples are given to illustrate the construction (Example \ref{ex1}). The case $p=2$ and $e=3$ is also discussed, where the general conditions reduce to a simpler form; see Theorem \ref{th1e=3p=2} and Example \ref{ex2}.

The paper is organized as follows. In Section \ref{sec22}, we present our general result together with its consequences. Section \ref{secc=-1} is devoted to the special case $c=-1$. In Section \ref{e3}, we address the case $e=3$. Section \ref{sec5} concludes the paper and states some open problems.
\section{Genaral Results and Some Consequences}\label{sec22}
For convenience, throughout this paper, we adopt the definition of the $c$-differential uniformity $\Delta(f,c)$ from Definition \ref{defcdiff} and fix the following notation unless stated otherwise.
\begin{itemize}
\item $q$ is a prime power, $q-1=en$ with $e\geq3$ and $e|n$.
\item $\mathbb{F}_{q}$ is the finite field with $q$ elements.
\item The multiplicative group $\mathbb{F}_{q}^*=\mathbb{F}_{q}\setminus\{0\}=\langle\pi\rangle$, where $\pi$ is a primitive element of $\mathbb{F}_q$. 
\item $D_0=\langle\pi^e\rangle,\,\,D_\lambda=\pi^\lambda D_0,\,\,0\leq \lambda\leq e-1$.
\item $S:=\mathbb{F}_q\setminus\{0,-1\}=\bigcup\limits_{\mu=0}^{e-1}\bigcup\limits_{\lambda=0}^{e-1} D_{\mu\lambda}.$
\item $D_{\mu\lambda} := D_\mu \cap (D_\lambda -1)=\{x\in S\mid x\in D_\mu,\; x+1\in D_\lambda\},$.
\item All subscripts of the cosets $D_\lambda$ and $D_{\mu\lambda}$ are taken modulo $e$.
\item $\mathbb{Z}_e=\mathbb{Z}/e\mathbb{Z}=\{0,1,\cdots,e-1\}$ denotes the ring of integers modulo $e$.
\end{itemize} 

We now present our main general result on power functions with bounded $c$-differential uniformity.
\begin{theorem}\label{generalcons}
	Let $q=p^m$, where $p$ is 
	prime and $m\ge 1$. Suppose that $q-1=en$ for an integer $e\ge 3$ with $e\mid n$. Let $\pi$ be a primitive element of $\mathbb{F}_q$, so that the multiplicative group $\mathbb{F}_q^*=\mathbb{F}_q\setminus\{0\}=\langle\pi\rangle$. Define
	\[
	D_0=\langle\pi^e\rangle,\qquad D_\lambda=\pi^\lambda D_0,\quad 0\le \lambda\le e-1,
	\]
	which are the cosets of the subgroup $D_0$ in $\mathbb{F}_q^*$. These $D_\lambda$ ($0\le \lambda\le e-1$) are the classical cyclotomic classes in $\mathbb{F}_q$. Let $\epsilon=\pi^n$, which is an element of order $e$ in $\mathbb{F}_q$. Assume that $c\in\mathbb{F}_q\setminus\{0,1,\epsilon,\dots,\epsilon^{e-1}\}$ and the following two conditions hold:
	
		$(1)$ The elements $1-c\epsilon^\lambda$ for $0\le \lambda\le e-1$ lie in a common coset $D_i$ for some fixed index $i\in\{0,1,\dots,e-1\}$;
		
		$(2)$ The elements $1-\epsilon^\lambda$ for $1\le \lambda\le e-1$ lie in a common coset $D_j$ for some fixed index $j\in\{0,1,\dots,e-1\}$.

	Then for each integer $l$ with $1\le l\le e-1$ and $\gcd(l,e)=1$, set 
	\[
	d=ln+1=\frac{l(q-1)}{e}+1.
	\]
	Then the power function $f(x)=x^d$ is a permutation polynomial over $\mathbb{F}_q$, and its $c$-differential uniformity satisfies $\Delta(f,c)\le e$.
\end{theorem}
\begin{remark}
	For power functions $f(x)=x^d$, Equation (\ref{eq:c_diff_eq}) becomes $(x+a)^d-cx^d=b$ with $a\neq 0$, which is equivalent to $(y+1)^d-cy^d=a^{-d}b$ via the substitution $x=ay$. Thus we only consider the case $a=1$ without loss of generality.
\end{remark}
To prove our main result, we require the following known result on cyclotomic classes.
\begin{lemma}\label{cyc}\cite{Tstorer1967cyclotomy}
	Let $D_i$ be a classical cyclotomic class of order $e$. Then
	\begin{enumerate}
		\item[(i)] For any $a\in D_j$ with $0\le j<e$, we have $aD_i = D_{i+j\pmod{e}}$;
		\item[(ii)] For $2\nmid q$, $-1\in D_0$ if $n$ is even, and $-1\in D_{\frac{e}{2}}$ if $n$ is odd.
	\end{enumerate}
\end{lemma}
\begin{remark}\label{th1}
	\begin{enumerate}
		\item[(i)] Since $D_0$ is the unique subgroup of order $n$ in the cyclic group $\mathbb{F}_q^*$, conditions (1) and (2) are independent of the choice of the primitive element $\pi$.
		\item[(ii)] From $\epsilon=\pi^n$ and $e\mid n$, we have $\epsilon\in D_0=\langle\pi^e\rangle$, and hence $\epsilon^\lambda\in D_0$ for all $0\le \lambda\le e-1$. Combining this with condition (1), for any $\mu,\lambda\in\{0,1,\dots,e-1\}$, by Lemma \ref{cyc} we obtain
		\[
		\epsilon^\lambda-c\epsilon^\mu=\epsilon^\lambda\big(1-c\epsilon^{\mu-\lambda}\big)\in D_i,
		\]
	under condition (1) of Theorem \ref{generalcons}.	
	\end{enumerate}
\end{remark}
\textbf{Proof of Theorem \ref{generalcons}.}
From $d=ln+1$ and $e\mid n$, we have $\gcd(n,d)=\gcd(e,d)=1$. 
Thus $\gcd(q-1,d)=\gcd(en,d)=1$, and hence $f(x)=x^d$ is a permutation polynomial over $\mathbb{F}_q$.
For each $b\in\mathbb{F}_q$, let $N(f,c,b)$ denote the number of solutions $x\in\mathbb{F}_q$ to the equation
\begin{equation}\label{Deltafunc}
	(x+1)^d-cx^d=b.
\end{equation}
Our goal is to prove $\Delta(f,c,b)\le e$ for all $b\in\mathbb{F}_q$, , which implies $\Delta(f,c)\le e$.

First consider $x\in S:=\mathbb{F}_q\setminus\{0,-1\}$. We have the disjoint union
\[
S=\bigcup_{\mu=0}^{e-1}\bigcup_{\lambda=0}^{e-1} D_{\mu\lambda},
\]
where
\[
D_{\mu\lambda}:=D_\mu\cap(D_\lambda-1)=\{x\in S\mid x\in D_\mu,\; x+1\in D_\lambda\}.
\] 
Let $x\in D_{\mu\lambda}$. Then $x=\pi^\mu A$ and $x+1=\pi^\lambda B$ for some $A,B\in D_0=\langle \pi^e\rangle$, so $A^n=B^n=1$. Recall $\epsilon=\pi^n$ and $d=ln+1$, we obtain 
\[
(x+1)^d = (\pi^\lambda B)^{ln+1}=(\pi^n)^{l\lambda}\pi^\lambda B=\epsilon^{l\lambda}(x+1),
\quad
x^d = (\pi^\mu A)^{ln+1}=\epsilon^{l\mu}x.
\]
Substituting into (\ref{Deltafunc}), the equation becomes
\begin{equation}\label{Deltafunceps}
	\epsilon^{l\lambda}(x+1)-c\epsilon^{l\mu}x=b.
\end{equation}
Equivalently,
\begin{equation}\label{DeltafuncAB}
	b=(\epsilon^{l\lambda}-c\epsilon^{l\mu})x+\epsilon^{l\lambda}=(\epsilon^{l\lambda}-c\epsilon^{l\mu})(x+1)+c\epsilon^{l\mu}.
\end{equation}

By conditions (1), (2) of Theorem \ref{generalcons} and Remark \ref{th1}, we have
\[
\epsilon^{l\lambda}-c\epsilon^{l\mu}=\epsilon^{l\lambda}\big(1-c\epsilon^{l(\mu-\lambda)}\big)\in D_i
\]
for some fixed index $i\in\{0,1,\dots,e-1\}$. In particular, $\epsilon^{l\lambda}-c\epsilon^{l\mu}\neq 0$.

We now establish two key facts.

\textbf{Claim\,(I)}. For fixed $(\mu,\lambda)$, distinct elements $x\in D_{\mu\lambda}$ yield distinct values of $b$ via (\ref{Deltafunceps})
(equivalently, via (\ref{Deltafunc})).

Indeed, if $x_1,x_2\in D_{\mu\lambda}$ both satisfy (\ref{Deltafunceps}), then
\[
\bigl(\epsilon^{l\lambda}-c\epsilon^{l\mu}\bigr)(x_1-x_2)=0.
\]
Since the coefficient is nonzero, we get $x_1=x_2$.
Moreover, from (\ref{DeltafuncAB}) we may solve for $x$ and $x+1$:
\begin{equation}\label{func_x_x+1}
	x=\frac{b-\epsilon^{l\lambda}}{\epsilon^{l\lambda}-c\epsilon^{l\mu}}\in D_\mu,\,\,\,x+1=\frac{b-c\epsilon^{l\mu}}{\epsilon^{l\lambda}-c\epsilon^{l\mu}}\in D_\lambda.
\end{equation}
Thus, for each $b\in\mathbb{F}_q$, equation (\ref{Deltafunceps}) has at most one solution in each $D_{\mu\lambda}$.

\textbf{Claim\,(II)}. Suppose $a\in D_{\mu\lambda}$ and $a'\in D_{\mu'\lambda'}$ are two distinct solutions of (\ref{Deltafunceps}) for the same $b$, with $(\mu,\lambda)\neq(\mu',\lambda')$. Then necessarily $\lambda\neq\lambda'$ and $\mu\neq\mu'$. 

From (\ref{DeltafuncAB}),
\begin{equation}\label{eq:3.5}
	\begin{aligned}
		b&=(\epsilon^{l\lambda}-c\epsilon^{l\mu})a+\epsilon^{l\lambda}
		=(\epsilon^{l\lambda}-c\epsilon^{l\mu})(a+1)+c\epsilon^{l\mu}\\
		&=(\epsilon^{l\lambda'}-c\epsilon^{l\mu'})a'+\epsilon^{l\lambda'}
		=(\epsilon^{l\lambda'}-c\epsilon^{l\mu'})(a'+1)+c\epsilon^{l\mu'}.
	\end{aligned}
\end{equation}
If $\lambda=\lambda'$, then (\ref{eq:3.5}) gives
\[
\bigl(\epsilon^{l\lambda}-c\epsilon^{l\mu}\bigr)a
=\bigl(\epsilon^{l\lambda}-c\epsilon^{l\mu'}\bigr)a'.
\]
Since
\[
\frac{\epsilon^{l\lambda}-c\epsilon^{l\mu}}{\epsilon^{l\lambda}-c\epsilon^{l\mu'}}\in D_{i-i}=D_0,
\]
and $a\in D_\mu$, $a'\in D_{\mu'}$, it follows that $\mu=\mu'$, contradicting $(\mu,\lambda)\neq(\mu',\lambda')$. Hence $\lambda\neq\lambda'$.
Similarly, if $\mu=\mu'$, then (\ref{eq:3.5}) yields $\lambda=\lambda'$, again a contradiction. Therefore $\mu\neq\mu'$. 

Now suppose equation (\ref{Deltafunceps}) has solutions in $t$ distinct sets
\[
D_{\mu_1\lambda_1},\; D_{\mu_2\lambda_2},\;\dots,\; D_{\mu_t\lambda_t}.
\]
By \textbf{Claim\,(II)}, the indices $\lambda_1,\lambda_2,\dots,\lambda_t$ are pairwise distinct, and each lies in $\{0,1,\dots,e-1\}$.
Hence $t\le e$. Consequently, for every $b\in\mathbb{F}_q$, the equation has at most $e$ solutions in $S$. It remains to handle the exceptional points $x=0$ and $x=-1$.
\begin{itemize}
	\item If $x=0$, equation (\ref{Deltafunc}) gives
	\[
	b=1^d - c\cdot 0^d = 1.
	\]
	\item If $x=-1$, equation (\ref{Deltafunc}) gives
	\[
	b=0^d - c(-1)^d = -c(-1)^d = (-1)^{d+1}c.
	\]
\end{itemize}

Thus, for any
\[
b\notin \bigl\{1,\; (-1)^{d+1}c\bigr\},
\]
we already have $\Delta(f,c,b)\le e$.

We now analyze the two exceptional values separately.

\textbf{Case 1:} $b=1$ and $1\neq (-1)^{d+1}c$.

Here $x=0$ is a solution, while $x=-1$ is not. Suppose $a\in D_{\mu\lambda}$ is a solution in $S$. From (\ref{DeltafuncAB}), we get
\[
1=(\epsilon^{l\lambda}-c\epsilon^{l\mu})a+\epsilon^{l\lambda}
=(\epsilon^{l\lambda}-c\epsilon^{l\mu})(a+1)+c\epsilon^{l\mu}.
\]
Solving for $a$ and $a+1$, we obtain
\[
a=\frac{1-\epsilon^{l\lambda}}{\epsilon^{l\lambda}-c\epsilon^{l\mu}},\qquad
a+1=\frac{1-c\epsilon^{l\mu}}{\epsilon^{l\lambda}-c\epsilon^{l\mu}}.
\]
By condition (2), $1-\epsilon^{l\lambda}\in D_j$; by condition (1), $\epsilon^{l\lambda}-c\epsilon^{l\mu}\in D_i$. Hence
\[
a\in D_{j-i},\qquad a+1\in D_{i-i}=D_0.
\]
Since $a\in D_\mu$ and $a+1\in D_\lambda$, we must have
\[
(\mu,\lambda)=(j-i,0).
\]
By \textbf{Claim\,(I)}, there is at most one such solution in $S$. Therefore
\[
\Delta(f,c,1)\le 1\ (\text{from }x=0)+1\ (\text{from }S)=2< e.
\]

\textbf{Case 2:} $b=(-1)^{d+1}c$. 

Here $x=-1$ is a solution, while $x=0$ is not. Suppose $a\in D_{\mu\lambda}$ is a solution in $S$. From (\ref{DeltafuncAB}), we have
\[
(-1)^{d+1}c
=(\epsilon^{l\lambda}-c\epsilon^{l\mu})a+\epsilon^{l\lambda}
=(\epsilon^{l\lambda}-c\epsilon^{l\mu})(a+1)+c\epsilon^{l\mu},
\]
Thus
\[
a=\frac{-\epsilon^{l\lambda}\big(1-(-1)^{d+1} c\epsilon^{-l\lambda}\big)}{\epsilon^{l\lambda}-c\epsilon^{l\mu}},\quad
a+1=\frac{(-1)^{d+1}c\big(1-(-1)^{d+1}\epsilon^{l\mu}\big)}{\epsilon^{l\lambda}-c\epsilon^{l\mu}}.
\]
- If $p$ is an odd prime, then $n$ is even, so 
 $2\nmid d$, then $(-1)^{d+1}=1$. Thus 
 \[
 a=\frac{-\epsilon^{l\lambda}\big(1- c\epsilon^{-l\lambda}\big)}{\epsilon^{l\lambda}-c\epsilon^{l\mu}},\quad
 a+1=\frac{c\big(1-\epsilon^{l\mu}\big)}{\epsilon^{l\lambda}-c\epsilon^{l\mu}}.
 \]
 - If $p=2$, so $(-1)^{d+1}=1$ as well, and the same expression holds.
 Let $c\in D_s$. Then
\[
a+1=\frac{c(1-\epsilon^{l\mu})}{\epsilon^{l\lambda}-c\epsilon^{l\mu}}\in D_{s+j-i},
\]
Since $a+1\in D_\lambda$,  it follows that $\lambda=s+j-i$ is uniquely determined. Suppose $-1\in D_t$ for some fixed index $t$, 
 one checks that $a\in D_t$, so $\mu=t$ is fixed. Then Again $(\mu,\lambda)$ is uniquely determined, giving at most one solution in $S$. 
Together with $x=-1$, Thus 

\[
N\bigl(f,c,(-1)^{d+1}c\bigr)\le 1\ (\text{from }x=-1)+1\ (\text{from }S)=2< e.
\]


Combining all \textbf{cases}, we have shown that for every $b\in\mathbb{F}_q$,
$N(f,c,b)\le e.$
Therefore the $c$-differential uniformity satisfies
\[
\Delta(f,c)=\max_{b\in\mathbb{F}_q}N(f,c,b)\le e.
\]

This completes the proof of Theorem \ref{generalcons}.

$\hfill\Box$
 \begin{remark}
	When $p$ is an odd prime and $e = 2$, the conditions reduce to the requirement that $1-c$ and $1+c$ lie in the same quadratic cyclotomic class, and the conclusion of Theorem \ref{generalcons} specializes to $\Delta(f,c) \le 2$, which agrees with the results in \cite{Mesnager2021tit}. 
\end{remark}

	The following result shows that conditions (1) and (2) in Theorem \ref{generalcons} are not restrictive.
\begin{theorem}\label{generalnot}
	For any given prime $p$ and integer $e\ge 3$ with $p\nmid e$. 
	
	$(1)$ There exist infinitely many extension fields $\mathbb{F}_q$ such that $e\mid q-1$ and $q\geq e+2$. Moreover, $\mathbb{F}_q^*$ contains an element $\epsilon$ of order $e$.
	
	$(2)$ Let $\mathbb{F}_q$ be any field from (1). Then there exist infinitely many extension fields $\mathbb{F}_Q$ of $\mathbb{F}_q$ such that, for every $c\in\mathbb{F}_q\setminus\{0,1,\epsilon,\dots,\epsilon^{e-1}\}$, the conditions (1) and (2) in Theorem \ref{generalcons} hold over $\mathbb{F}_Q$ with $i=j=0$. Consequently, for infinitely many such fields $\mathbb{F}_Q$, writing $Q-1=eN$ with $e\mid N$. Then for each $d=lN+1$, where $1\le l\le e-1$ and $\gcd(l,e)=1$, the power function $f(x)=x^d$ defined over $\mathbb{F}_Q$ has $c$-differential uniformity $\Delta(f,c)\le e$.
\end{theorem}
{\bf Proof.} $(1)$ Since $p\nmid e$, $p$ is an invertible element in the ring $\mathbb{Z}_e=\mathbb{Z}/e\mathbb{Z}$. Let $r$ denote the multiplicative order of $p$ modulo $e$, i.e., the smallest positive integer satisfying
\[
p^r\equiv 1\pmod{e}.
\]
Set $q'=p^r$ and $K=\mathbb{F}_{q'}$. Since $e\mid q'-1$, $K^*$ contains a unique cyclic subgroup of order $e$. Take a primitive element $\pi'\in K^*$. Then
\[
\epsilon=(\pi'
)^{\frac{p^r-1}{e}}
\]
is an element of order $e$ in $K$. Moreover, every element of order $e$ in the algebraic closure $\overline{\mathbb{F}}_p$ belongs to $K$. 
Now let $\mathbb{F}_q$ be any extension of $K$, i.e., $q=(q')^l$ for some integer $l\ge 1$. Then 
\[
q-1=(q')^l-1=(q'-1)(1+q'+\cdots+(q')^{l-1}),
\]
so $e\mid q-1$. Since $q'\geq e+1$, taking $l\geq 2$ ensures $q=(q')^l\ge e+2$. Since infinitely many positive integers $l$ satisfy this condition, there exist infinitely many such fields $\mathbb{F}_q$. Note that $\epsilon$ is an element of order $e$ in $K=\mathbb{F}_{q'}$, and $K\subseteq\mathbb{F}_q$, so $\epsilon\in\mathbb{F}_q$. 

$(2)$ 
Let $\mathbb{F}_q$ be a finite field obtained from part (1). Consider the extension field $\mathbb{F}_Q$ of $\mathbb{F}_q$ with $Q=q^{em}$, where $m\ge 1$. Let $\pi$ be a primitive element of $\mathbb{F}_Q$, so that $\mathbb{F}_Q^*=\langle\pi\rangle$. Since $\mathbb{F}_q\subseteq\mathbb{F}_Q$, the multiplicative group $\mathbb{F}_q^*$ is the unique subgroup of $\mathbb{F}_Q^*$ of order $q-1$. Therefore
\[
\mathbb{F}_q^*=\left\langle \pi^{\frac{Q-1}{q-1}} \right\rangle.
\]
Set
\[
\tilde{\pi}= \pi^{\frac{Q-1}{q-1}}.
\]
We claim that $\tilde{\pi}\in \langle \pi^e\rangle = D_0$.
Since $q\equiv 1\pmod{e}$, we compute
\[
\gcd\left(e,\frac{Q-1}{q-1}\right)
=\gcd\left(e,\frac{q^{em}-1}{q-1}\right)
=\gcd\big(e,1+q+q^2+\cdots+q^{em-1}\big).
\]
Reducing each term modulo $e$, we have 
\[
\gcd\big(e,1+q+q^2+\cdots+q^{em-1}\big)
=\gcd(e,\underbrace{1+1+\cdots+1}_{em\ \text{terms}})=\gcd(e,em)=e.
\]
Hence $e\mid \frac{Q-1}{q-1}$, which implies that $\tilde{\pi}= \pi^{\frac{Q-1}{q-1}}\in D_0$. 
 It follows that $\mathbb{F}_q^*=\langle \tilde{\pi}\rangle \subseteq D_0$. 
By part (1), $\mathbb{F}_q$ contains an element $\epsilon$ of order $e$. Now take any
\[
c\in \mathbb{F}_q\setminus \{0,1,\epsilon,\dots,\epsilon^{e-1}\}.
\]
For $0\le \lambda\le e-1$, we have 
$1-c\epsilon^\lambda\in\mathbb{F}_q^*$, so $1-c\epsilon^\lambda\in D_0$.
Similarly, for $1\le \lambda\le e-1$, $\epsilon^\lambda\neq 1$, 
we have $1-\epsilon^\lambda\in\mathbb{F}_q^*\subseteq D_0$. Thus conditions (1) and (2) of Theorem \ref{generalcons} hold over $\mathbb{F}_Q$ with $i=j=0$.

Finally, recall $Q=q^{em}$. Since $e\mid q-1$ and $e\mid \frac{Q-1}{q-1}$, we get $e^2\mid Q-1$. Write $Q-1=eN$, then $e\mid N$. The last statement then follows directly from Theorem \ref{generalcons}.

$\hfill\Box$

Theorem \ref{generalcons} provides a construction of power function $f(x)=x^d$ over finite fields $\mathbb{F}_q$ with $q=p^m$, $q-1=en$ and $e\mid n$. 
 Under conditions (1) and (2) of Theorem \ref{generalcons}, for $d=ln+1$ ($1\leq l\leq e-1, \gcd(l,e)=1$) and $c\in \mathbb{F}_q\setminus \{0,1,\epsilon,\dots,\epsilon^{e-1}\}$, where $\epsilon$ is an element of order $e$, then $\Delta(f,c)\leq e$.

A natural question arises: is the bound $\Delta(f,c)\leq e$ tight? The following result (Theorem \ref{th1tight}) gives an affirmative answer. To prove it, we shall need a classical result on character sums.

\begin{lemma}\label{weil}{\rm(\cite{mullen2013handbook}, Theorem 6.2.2)}
Let $\mathbb{F}_q$ be a finite field with multiplicative group $\mathbb{F}_q^*=\langle \pi\rangle$, and let $e$ be a positive integer dividing $q-1$.
Let $\psi$ be a multiplicative character of $\mathbb{F}_q^*=\langle \pi\rangle$ of order $e$ defined by $\psi(\pi^i)=\zeta_e^i$ for $0\le i\le q-2$ and $\psi(0)=0$, where $\zeta_e=\exp(\frac{2\pi \sqrt{-1}}{e})$. 


	
	Let $f\in\mathbb{F}_q[x]$ be a polynomial of degree $d\ge 1$.
		If $f$ is not an $e$-th power in $\overline{\mathbb{F}}_q[x]$, where $\overline{\mathbb{F}}_q$ denotes the algebraic closure of $\mathbb{F}_q$, then
	\[
	\left|\sum_{x\in\mathbb{F}_q}\psi(f(x))\right|\le (d-1)\sqrt{q}.
	\]

	\end{lemma}
	\begin{theorem}\label{th1tight}
		Let $e\geq3$. Let $\mathbb{F}_q$ and the power function $f(x)=x^d$ over $\mathbb{F}_q$ be as constructed in Theorem \ref{generalcons}. That is, $q-1=en$ with $e\mid n$, $\mathbb{F}_q^*=\langle \pi\rangle$, $D_0=\langle \pi^e\rangle$, $D_\lambda=\pi^\lambda D$ for $0\leq\lambda\leq e-1$, and $\epsilon=\pi^n$ is an element of order $e$ in $\mathbb{F}_q$. For $d=ln+1$, where $1\leq l\leq e-1$,  $\gcd(l,e)=1$, and $c\in \mathbb{F}_q\setminus \{0,1,\epsilon,\dots,\epsilon^{e-1}\}$ satisfies conditions (1) and (2) of Theorem \ref{generalcons}.
		If $q-(e^{2e}-1)(2e(e-1)-1)\sqrt{q}\geq1$, then $c$-differential uniformity of $f$ is exactly $\Delta(f,c)=e$.
		\end{theorem}
		\begin{remark}\label{tightRemark}
			
			By Theorem \ref{generalnot}, there exist infinitely many extension fields of $\mathbb{F}_q$ that also satisfy the construction of Theorem \ref{generalcons}. Hence we may choose $q$ sufficiently large so that the inequality of Theorem \ref{th1tight} holds. Moreover,  the condition $q-1=en$ with $e\mid n$ implies $q\ge e^2+1$, so
			\[
			\big|\mathbb{F}_q\setminus\{0,1,\epsilon,\dots,\epsilon^{e-1}\}\big| = q-(e+1)\ge e(e-1)\ge 1.
			\] 
			Thus there exists at least one $c \in \mathbb{F}_q \setminus \{0,1,\epsilon,\dots,\epsilon^{e-1}\}$. 
		\end{remark}
		\textbf{Proof of Theorem \ref{th1tight}.} 
Let $D_{\mu\lambda}:=D_\mu\cap(D_\lambda-1)$ ($0\le \mu,\lambda\le e-1$). 
 For $b\in \mathbb{F}_q$ and $c \in \mathbb{F}_q \setminus \{0,1,\epsilon,\dots,\epsilon^{e-1}\}$, we define
\[
N(b)=
\begin{cases}
	1, & \text{if there exist $e$ distinct solutions } a_1,\dots,a_e
		 \text{ of the equation } \\
	& (x+1)^d - c x^d = b, \text{ with } a_k\in D_{\mu_k\lambda_k} \text{ for each } 1\le k\le e,\\
	& \text{where } \{\mu_1,\dots,\mu_e\} \text{ and } \{\lambda_1,\dots,\lambda_e\} \text{ are permutations of } \{0,1,\dots,e-1\};\\[4pt]
	0, & \text{otherwise}.
\end{cases}
\]
We will show that $M=\sum\limits_{b\in \mathbb{F}_q}N(b)\geq1$.
This implies that there exists some $b\in\mathbb{F}_q$ satisfying $N(b)=1$. By definition of $N(b)$, the equation $(x+1)^d - c x^d = b$ possesses $e$ distinct solutions $x=a_k$, which gives $\Delta(f,c)\ge e$. Combining this lower bound with the upper bound $\Delta(f,c)\le e$ from Theorem \ref{generalcons}, we obtain $\Delta(f,c)=e$. 

We have that 
\begin{align*}
	&\,\, a_k \in D_{\mu_k\lambda_k}\,\;(1\le k\le e) \text{ are solutions of } (x+1)^d - c x^d = b\\
	&\iff a_k\in D_{\mu_k\lambda_k} \text{ and } \epsilon^{l{\lambda_k}}(a_k+1)-c\epsilon^{l{\mu_k}}a_k=b\,\;(1\le k\le e) \\
	&\iff a_k\in D_{\mu_k\lambda_k} \text{ and } b= (\epsilon^{l{\lambda_k}}-c\epsilon^{l{\mu_k}})a_k+\epsilon^{l{\lambda_k}}=(\epsilon^{l{\lambda_k}}-c\epsilon^{l{\mu_k}})(a_k+1)+c\epsilon^{l{\mu_k}}\,\;(1\le k\le e)\\
	&\iff a_k=\frac{b-\epsilon^{l{\lambda_k}}}{\epsilon^{l{\lambda_k}}-c\epsilon^{l{\mu_k}}}\in D_{\mu_k} \text{ and } a_k+1=\frac{b-c\epsilon^{l{\mu_k}}}{\epsilon^{l{\lambda_k}}-c\epsilon^{l{\mu_k}}}\in D_{\lambda_k}\,\;(1\le k\le e)\\
	&\iff b-\epsilon^{l{\lambda_k}} \in D_{\mu_k+i} =: D_{\alpha_k} \text{ and } b-c\epsilon^{l{\mu_k}} \in D_{\lambda_k+i} =: D_ {\beta_k}\,\;(1\le k\le e).
\end{align*}
The last equivalence follows from Remark \ref{th1} (ii), which gives $\epsilon^{l\lambda_k}-c\epsilon^{l\mu_k}\in D_i$. Here the indices $\alpha_k$ and $\beta_k$ are determined by $\alpha_k\equiv \mu_k+i \pmod{e}$ and $\beta_k\equiv \lambda_k+i\pmod{e}$ for $1\le k\le e$. 

Let $\psi$ be a multiplicative character of $\mathbb{F}_q^*=\langle \pi\rangle$ defined by $\psi(\pi^i)=\zeta_e^i$ for $0\le i\le q-2$ and $\psi(0)=0$, where $\zeta_e=\exp(\frac{2\pi \sqrt{-1}}{e})$. By orthogonality relation of multiplicative characters, for any $a\in \mathbb{F}_q^*$, 
\[
\frac{1}{e}\sum\limits_{\tau=0}^{e-1}\psi^\tau(\pi^{-i}a)=
\begin{cases}
	1, & \text{if } a\in D_i;\\
	0, & \text{otherwise}.
\end{cases}
\]
Therefore, 
\[
\begin{aligned}
	M &= \sum_{b\in\mathbb F_q} N(b)
	= \sum_{\substack{b\in\mathbb F_q,\\ b-\epsilon^{l\lambda_k}\in D_{\alpha_k},\; b-c\epsilon^{l\mu_k}\in D_{\beta_k},\\ 1\le k\le e}} 1 \\
	&= \frac{1}{e^{2e}}\sum_{b\in\mathbb F_q}
	\left( \prod_{k=1}^e \sum_{\tau_k=0}^{e-1} \psi^{\tau_k}\big((b-\epsilon^{l\lambda_k})\,\pi^{-\alpha_k}\big) \right)
	\left( \prod_{k=1}^e \sum_{\sigma_k=0}^{e-1} \psi^{\sigma_k}\big((b-c\epsilon^{l\mu_k})\,\pi^{-\beta_k}\big) \right) \\
	&= \frac{1}{e^{2e}}\sum_{b\in\mathbb F_q}
	\left( \sum_{\tau_1,\dots,\tau_e=0}^{e-1} \zeta_e^{-\sum\limits_{k=1}^e \tau_k\alpha_k} \psi\Big(\prod_{k=1}^e (b-\epsilon^{l\lambda_k})^{\tau_k}\Big) \right)
	\left( \sum_{\sigma_1,\dots,\sigma_e=0}^{e-1} \zeta_e^{-\sum\limits_{k=1}^e \sigma_k\beta_k} \psi\Big(\prod_{k=1}^e (b-c\epsilon^{l\mu_k})^{\sigma_k}\Big) \right) \\
	&= \frac{1}{e^{2e}} \left(
	q + \sum_{\substack{\tau_1,\dots,\tau_e,\sigma_1,\dots,\sigma_e=0\\ \tau_1+\cdots+\tau_e+\sigma_1+\cdots+\sigma_e \ge 1}}^{e-1}
	\zeta_e^{-\left(\sum\limits_{k=1}^e \tau_k\alpha_k+\sum\limits_{k=1}^e \sigma_k\beta_k\right)}
	\sum_{b\in\mathbb F_q}\psi\big(f_{\bar{\tau},\bar{\sigma}}(b)\big)
	\right),
\end{aligned}
\]
where 
\begin{equation*}
	\begin{cases}
		\bar{\tau}=(\tau_1,\cdots,\tau_e),\,\,\bar{\sigma}=(\sigma_1,\cdots,\sigma_e),\,\,\bar{\alpha}=(\alpha_1,\cdots,\alpha_e),\,\,\bar{\beta}=(\beta_1,\cdots,\beta_e);\\
	\bar{\tau}\cdot\bar{\alpha}=\sum\limits_{k=1}^{e}\tau_k\alpha_k,\,\,\bar{\sigma}\cdot\bar{\beta}=\sum\limits_{k=1}^{e}\sigma_k\beta_k;\\
	f_{\bar{\tau},\bar{\sigma}}(x)=\prod\limits_{k=1}^{e}((x-\epsilon^{l\lambda_k})^{\tau_k}(x-c\epsilon^{l\mu_k})^{\sigma_k}).
	\end{cases}
\end{equation*}

From $c\notin \{0,1,\epsilon,\dots,\epsilon^{e-1}\}$ and $\gcd(l,e)=1$, the $2e$ elements $\epsilon^{l\lambda_k}$ and $c\epsilon^{l\mu_k}$ for $1 \le k \le e$ are pairwise distinct. Hence the linear factors $x-\epsilon^{l\lambda_k}$ and $x-c\epsilon^{l\mu_k}$ are pairwise distinct. Let $d = \tau_1+\cdots+\tau_e+\sigma_1+\cdots+\sigma_e \ge 1$. Since $0 \le \tau_k,\sigma_k \le e-1$, we have $1\le \deg(f_{\bar \tau,\bar \sigma}(x))=d\le 2e(e-1)$. Moreover, $f_{\bar \tau,\bar \sigma}(x)$ is not an $e$-th power in $\overline{\mathbb F}_q[x]$, 
 by Lemma \ref{weil} we have 
\[
\left|\sum_{x\in\mathbb F_q}\psi(f_{\bar \tau,\bar \sigma}(x))\right|
\le (d-1)\sqrt{q} \le \big(2e(e-1)-1\big)\sqrt{q}.
\]
Thus
\[
\begin{aligned}
	M &\ge \frac{1}{e^{2e}} \biggl( q - \sum_{\substack{\tau_1,\dots,\tau_e,\sigma_1,\dots,\sigma_e=0\\ \tau_1+\cdots+\tau_e+\sigma_1+\cdots+\sigma_e \ge 1}}^{e-1}
	\Bigl| \zeta_e^{-\bar{\tau}\cdot\bar{\alpha}-\bar{\sigma}\cdot\bar{\beta}} \Bigr|
	\cdot \biggl| \sum_{b\in\mathbb{F}_q}\psi\bigl(f_{\bar{\tau},\bar{\sigma}}(b)\bigr) \biggr| \biggr) \\
	&\ge \frac{1}{e^{2e}} \Bigl( q - (e^{2e}-1)\bigl(2e(e-1)-1\bigr)\sqrt{q}\Bigr) \\
	&\ge \frac{1}{e^{2e}} > 0,
\end{aligned}
\]
by the assumed inequality in Theorem \ref{th1tight}. From the argument at the beginning of this proof, we get $\Delta(f,c)=e$.

$\hfill\Box$

\section{The case $c=-1$}\label{secc=-1}
In this section, we give further discussion for $c=-1$, one of the popular cases. 
 Recall the equation $(x+1)^d - cx^d = b$.
Let $\epsilon$ be a primitive $e$-th root of unity, so $\epsilon^e=1$.
If $e$ were even, then $-1=\epsilon^{e/2}$ belongs to $\{1,\epsilon,\dots,\epsilon^{e-1}\}$.
The requirement $c=-1\notin\{0,1,\epsilon,\dots,\epsilon^{e-1}\}$ thus implies that $e$ must be odd, so $2\nmid e$. 
Let $q$ be a power of an odd prime $p$.
Write $q-1=en$ with $e\mid n$. Since $q-1$ is even and $e$ is odd, $n$ must be even, i.e., $2\mid n$.
It follows that $d=ln+1$ is odd, $2\nmid d$. 
The conditions (1) and (2) of Theorem \ref{generalcons} become:

	$(B1)$ The elements $1-c\epsilon^\lambda=1+\epsilon^\lambda$ for $0\le \lambda\le e-1$ belong to a common coset $D_i$ for some index $i\in\{0,1,\dots,e-1\}$;

$(B2)$ The elements $1-\epsilon^\lambda$ for $1\le \lambda\le e-1$ belong to a common coset $D_j$ for some index $j\in\{0,1,\dots,e-1\}$.

For $1\le \lambda\le e-1$, we have  $1+\epsilon^\lambda=\frac{1-\epsilon^{2\lambda}}{1-\epsilon^\lambda}\in D_{j-j}=D_0$, which implies that condition (B1) holds for all $0\le \lambda\le e-1$ with $i=0$. For $\lambda=0$, $1+\epsilon^0=2$. Thus condition (B1) holds for all $0\le \lambda\le e-1$ if and only if additionally $2\in D_0$. Therefore we get the following result:
\begin{theorem}\label{th1case-1}
	Let $2\nmid e\geq3$. Let $q=p^m$ with $p\geq 3$ an odd prime, $q-1=en$ with $e\mid n$, $\mathbb{F}_q^*=\langle \pi\rangle$, $D_0=\langle \pi^e\rangle$, $D_\lambda=\pi^\lambda D_0$ for $0\leq\lambda\leq e-1$, and $\epsilon=\pi^n$ is an element of order $e$ in $\mathbb{F}_q$, Assume the following condition holds:
	
	$(B)$ $2\in D_0$, and the elements $1-\epsilon^\lambda$ for $1\le \lambda\le e-1$ belong to a common coset $D_j$ for some index $j\in\{0,1,\dots,e-1\}$.
	
	Then for $c=-1$ and for every $d=ln+1$ with $1\leq l\leq e-1$ and  $\gcd(l,e)=1$, the power function $f(x)=x^d$ over $\mathbb{F}_q$ has ($-1$)-differential uniformity $\Delta(f,-1)\leq e$.
\end{theorem}

Let $d\ge 1$ be an odd integer. For $b\in \mathbb{F}_q$, we denote by $N(b)=N(f,c,b,1)$ the number of solutions in $\mathbb{F}_q$ to the equation $(x+1)^d - cx^d = b$. 
For the special case $c=-1$ and $a\in\mathbb{F}_q$, the quantity $\boldsymbol{\Delta f}(a)=(a+1)^d+a^d$ is referred to as the $\boldsymbol{\Delta f}$‑value of $a$. Some properties concerning $N(b)$ and these $\boldsymbol{\Delta f}$‑values will be given below. 
\begin{lemma}\label{c-1lemma}
	Let $N(b)$ and the $\boldsymbol{\Delta f}$‑value be defined above. Assume that all conditions in Theorem \ref{th1case-1} hold. Then
	
	$(1)$ For each $a\in \mathbb{F}_q$, $a\in D_{\mu\lambda}\iff -(a+1)\in D_{\lambda\mu}$; $\boldsymbol{\Delta f}(a)=b\iff \boldsymbol{\Delta f}(-(a+1))=-b$.
	
	$(2)$ For $b\in \mathbb{F}_q \setminus \{\pm1\}$, $N(b)=N(-b)\leq e$.
	
	 $(3)$ $N(0)=1$ and $N(1)=N(-1)=1$.
	\end{lemma}
	{\bf Proof.} $(1)$ Let $\phi\colon \mathbb{F}_q\to\mathbb{F}_q$ be given by $\phi(a)=-(a+1)$. $\phi$ is an involution since $\phi(\phi(a))=\phi (-(a+1))=-(-(a+1)+1)=a$. 
	For $2\nmid q=p^m$, $\mathbb{F}_q$ is partitioned into $\frac{q-1}{2}$ pairs $\{a,-(a+1)\}$ and one element $-\frac{1}{2}$. Moreover, since $2\nmid d$, for any $a\in \mathbb{F}_q$, we have 
	\begin{eqnarray*}\nonumber
		\boldsymbol{\Delta f}(-(a+1))
		&=&(-(a+1)+1)^d+(-(a+1))^d\\
		&=&-(a)^d-(a+1)^d\\
		&= &-\boldsymbol{\Delta f}(a).
	\end{eqnarray*}
	Thus, $\boldsymbol{\Delta f}(a)=b\iff \boldsymbol{\Delta f}(-(a+1))=-b$.
	
	 Recall that 
	$D_{\mu\lambda}:=D_\mu\cap(D_\lambda-1)$, we get
	\begin{align*}
		a\in D_{\mu\lambda}
		&\iff a\in D_\mu,\ a+1\in D_\lambda \\
		&\iff -a\in D_\mu,\ -(a+1)\in D_\lambda\,\,\,(\text{since}\,2\mid n,\,-1=\pi^\frac{en}{2}\in D_0) \\
		&\iff -(a+1)\in D_\lambda,\quad \big(-(a+1)+1\big)=-a\in D_\mu \\
		&\iff -(a+1)\in D_{\lambda\mu}.
	\end{align*}
	
	$(2)$ The statement (2) is an immediate consequence of (1).
	
    $(3)$ If $\boldsymbol{\Delta f}(a)=0$, then $(a+1)^d+a^d=0$. Clearly  
    $a\neq 0, -1$. Hence $(\frac{a}{a+1})^d=-1$. From $\gcd(q-1,d)=1$, we get $\frac{a}{a+1}=-1$ which gives $a=-\frac{1}{2}$. Thus $N(0)=1$.
    
    By (2) we have $N(b)=N(-b)$ for all $b\in \mathbb{F}_q$. Therefore it suffices to prove $N(1)=1$ as this will imply $N(1)=N(-1)=1$. 
    
    We first note that $x=0$ is a solution of  $(x+1)^d+x^d=1$. Suppose, for contradiction, that there exists a solution $a \in S=\bigcup\limits_{\mu=0}^{e-1}\bigcup\limits_{\lambda=0}^{e-1} D_{\mu\lambda}=\mathbb{F}_q \setminus \{0,-1\}$. Then $a\in D_{\mu\lambda}$ for some $\mu, \lambda$, i.e., $a\in D_\mu$ and $a+1\in D_\lambda$. 
   Recall that $d=l n+1$ with $\gcd(l,e)=1$. Substituting into $(a+1)^d+a^d=1$ gives
    \[
    1=(\epsilon^{l\lambda}+\epsilon^{l\mu})a+\epsilon^{l\lambda}
    =(\epsilon^{l\lambda}+\epsilon^{l\mu})(a+1)-\epsilon^{l\mu}.
    \]
    Note that $\epsilon^{l\lambda}+\epsilon^{l\mu}\neq 0$ since $e$ is odd and the order of $\epsilon$ is $e$. Therefore 
    \begin{equation}\label{a+1D}
   	a+1=\frac{1+\epsilon^{l\mu}}{\epsilon^{l\mu}+\epsilon^{l\lambda}}\in D_\lambda.
   \end{equation}
   
    On the other hand, from condition (B) we know that $1+\epsilon^0=2\in D_0$ and for $1\le \mu\le e-1$, $1+\epsilon^{l\mu}=\frac{1-\epsilon^{2l\mu}}{1-\epsilon^{l\mu}}\in D_{j-j}=D_0$. Therefore $$\epsilon^{l\mu}+\epsilon^{l\lambda}=\epsilon^{l\mu}(1+\epsilon^{l\lambda-l\mu})\in D_0,\,\,0\le\mu, \lambda\le e-1.$$ 
    Thus
     \begin{equation}\label{a+1D0}
    	a+1=\frac{1+\epsilon^{l\mu}}{\epsilon^{l\mu}+\epsilon^{l\lambda}}\in D_0.
    \end{equation}
     Since distinct cosets are disjoint, from (\ref{a+1D}) and  (\ref{a+1D0}) we must have $\lambda=0$. Then 
    \[
    1=(\epsilon^{l\lambda}+\epsilon^{l\mu})a+\epsilon^{l\lambda}
    =(1+\epsilon^{l\mu})a+1.
    \]
    Hence $(1+\epsilon^{l\mu})a=0$. Since $a\neq 0$, we must have $1+\varepsilon^{l\mu}=0$, i.e., $\varepsilon^{l\mu}=-1$. However, $e$ is odd, so $(-1)^e=-1\neq 1$, which implies that $-1$ cannot be an $e$-th root of unity. This contradicts the fact that $\varepsilon^{l\mu}$ is an $e$-th root of unity.
    Therefore, no solution exists in $\mathbb{F}_q\setminus\{0,-1\}$, and $x=0$ is the unique solution of $(x+1)^d+x^d=1$. Thus $N(1)=1$, and consequently $N(-1)=1$.
    
    $\hfill\Box$

\section{The case $e=3$ }\label{e3} 
We now consider the special case $e=3$. Let $q=p^m=3n+1$ with $3\mid n$ (so $q\equiv 1 \pmod{9}$). Let $\mathbb{F}_q^*=\langle \pi\rangle$, $D_0=\langle \pi^3\rangle$, $D_\lambda=\pi^\lambda D_0$ for $0\leq\lambda\leq 2$, and let $\epsilon=\pi^n$ is an element of order $3$ in $\mathbb{F}_q$. Then $1+\epsilon+\epsilon^2=0$. From $3\mid n$ we get $\epsilon, \epsilon^2\in D_0$. For $c\in \mathbb{F}_q\setminus \{0,1,\epsilon,\epsilon^{2}\}$, conditions (1) and (2) of Theorem \ref{generalcons} become:

$(C1)$ $1-c\epsilon^\lambda$ ($0\leq\lambda\leq 2$) belong to a common coset $D_i$ for some $0\leq i\leq 2$;

$(C2)$ $1-\epsilon$ and $1-\epsilon^2$ belong to a common coset $D_j$ for some $0\leq j\leq 2$.

Condition $(C2)$ implies $\frac{1-\epsilon^2}{1-\epsilon}=1+\epsilon=-\epsilon^2\in D_0$. Since $\epsilon^2\in D_0$, this is equivalent to $-1\in D_0$, which holds if and only if $2\mid n$.
We thus obtain the following result:
\begin{theorem}\label{th1e=3}
	Let $q=p^m$ with $p\geq3$ an odd prime, and suppose $q-1=3n$ with $6\mid n$. Let $\mathbb{F}_q^*=\langle \pi\rangle$, $D_0=\langle \pi^3\rangle$, $D_\lambda=\pi^\lambda D_0$ for $0\leq\lambda\leq 2$, and let $\epsilon=\pi^n$ is an element of order $3$ in $\mathbb{F}_q$. 
	For $c\in \mathbb{F}_q\setminus \{0,1,\epsilon,\epsilon^{2}\}$, assume the following condition holds:
	
	$(C)$ $1-c$, $1-c\epsilon$ and $1-c\epsilon^2$ belong to a common coset $D_i$ for some $0\leq i\leq 2$.
	
	Then for $d=n+1$ and $d=2n+1$, the power function $f(x)=x^d$ over $\mathbb{F}_q$ has $c$-differential uniformity $\Delta(f,c)\leq 3$.
	\end{theorem}
	
	Next, we determine how many 
	$c$ satisfy condition $(C)$. To this end, we need the following lemma.
	\begin{lemma}\label{clc}{\rm(\cite{Tstorer1967cyclotomy})}
		Suppose that $q=p^m\equiv 1 \pmod{3}$, $D_\lambda$ ($0\leq\lambda\leq 2$) are the classical cyclotomic classes of $\mathbb{F}_q^*$ of order $3$, Define $D_{00}=D_0\cap(D_0-1)=\{x\in \mathbb{F}_q\setminus\{0,-1\}|x\in D_0,\,x+1 \in D_0\}$. Then
		
		$(1)$ There exist integers $s$ and $t$ such that $4q=s^2+27t^2$ with $s\equiv 1 \pmod{3}.$ 
		If $p\equiv 1 \pmod{3}$, then $p\nmid st$. Moreover, $s$ is uniquely determined by $q$, and $t$ is unique up to sign.
		
		$(2)$ $|D_{00}|=\frac{1}{9}(q-8+s).$ 
		
		$(3)$ $2\in D_0$ if and only if $2\mid s$, which implies $2\mid t$ and 
		$q=\left(\frac{s}{2}\right)^2+27\left(\frac{t}{2}\right)^2.
		$
		\end{lemma}
		\begin{theorem}\label{cmany}
			Let the assumptions be as in Theorem \ref{th1e=3}, and suppose further that $q\geq 73$. Then the number of elements $c\in \mathbb{F}_q\setminus \{0,1,\epsilon,\epsilon^{2}\}$ satisfying condition ($C$) is 
			$\Gamma=\frac{1}{9}(q+s-26)$, where $s$ is given by Lemma \ref{clc}. 
	\end{theorem}
		{\bf Proof.} Condition ($C$) means that 
		\[
		\frac{1-c\epsilon}{1-c}\in D_0,\,\,\,\frac{1-c\epsilon^2}{1-c}\in D_0.
		\] It is well-known that 
		\[
		\varphi:\mathbb{F}_q\cup\{\infty\}\rightarrow\mathbb{F}_q\cup\{\infty\},\,\,\,\varphi(x)=\frac{1-x\epsilon }{1-x}.
		\] 
		is a permutation of the projective line $\mathbb{F}_q\cup\{\infty\}$ ($\infty=\frac{1}{0}$). We have 
		\[
		\varphi(0)=1,\,\,\varphi(1)=\infty,\,\,\varphi(\epsilon)=1+\epsilon=-\epsilon^2,\,\,\varphi(\epsilon^2)=0.
		\]
		Therefore,
	\[
\Gamma=\#\Bigl\{c\in \mathbb{F}_q\setminus \{0,1,\epsilon,\epsilon^{2}\}\Bigr|\,\,\frac{1-c\epsilon}{1-c}\in D_0,\,\,\,\frac{1-c\epsilon^2}{1-c}\in D_0\Bigr\}.
	\]	
	Thus 
		\begin{align*}
		\Gamma&=\#\Bigl\{c\in \mathbb{F}_q\setminus \{0,1,\epsilon,\epsilon^{2}\}\Bigr|\,\,\frac{1-c\epsilon}{1-c}\in D_0,\,\,\,\frac{1-c\epsilon^2}{1-c}\in D_0\Bigr\}\\
		&=\#\Bigl\{\alpha\in \mathbb{F}_q\setminus \{0,1,\epsilon,-\epsilon^{2}\}\Bigr|\,\,\alpha\in D_0,\,\,\,-\alpha\epsilon^2-\epsilon\in D_0\Bigr\}\,\,(\text{let } \alpha=\varphi(c)=\frac{1-c\epsilon}{1-c})\\
			&=\#\Bigl\{\beta\in \mathbb{F}_q\setminus \{0,\epsilon,\epsilon^{2},-1\}\Bigr|\,\,\epsilon^2\beta\in D_0,\,\,\,-\epsilon(\beta+1)\in D_0\Bigr\} \,\,(\text{let } \alpha=\epsilon^2\beta)\\
		&=\#\Bigl\{\beta\in \mathbb{F}_q\setminus \{0,\epsilon,\epsilon^{2},-1\}\Bigr|\,\,\beta\in D_0,\,\,\,\beta+1\in D_0\Bigr\} \,\,\\
		&=\#\Bigl\{\beta\in \mathbb{F}_q\setminus \{\epsilon,\epsilon^{2}\}\Bigr|\,\,\beta\in D_{00}\} \,\,(\text{since } 0\notin D_{00},\,-1\notin D_{00}).
	\end{align*}
	
	Since $\epsilon,\epsilon^2\in D_0$ and $\epsilon+1=-\epsilon^2\in D_0$, $\epsilon^2+1=-\epsilon\in D_0$, we get $\epsilon,\epsilon^2\in D_{00}$. Thus
	$$\Gamma=|D_{00}|-2=\frac{1}{9}(q+s-26)$$
	 by Lemma \ref{clc} (2). Finally, since $4q=s^2+27t^2$, we have $|s|\leq 2\sqrt{q}$. Then, when $q\geq73$, $$\Gamma\geq \frac{1}{9}(q+s-26)>0.$$ 
	Thus there exists $c\in \mathbb{F}_q\setminus \{0,1,\epsilon,\epsilon^{2}\}$ satisfying condition ($C$) in Theorem \ref{th1e=3}.
	
	$\hfill\Box$
		\begin{remark}\label{q=19or37} 
			Among $q<73$ satisfying $q\equiv 1\pmod{18}$, the only possibilities are $q=19$ and $q=37$.
			For $q=19$, we have $4q=76=7^2+27\cdot 1^2$, so $s=7$, and
			\[
			\Gamma=\frac{1}{9}(19+7-26)=0.
			\]
			For $q=37$, we have $4q=148=(-11)^2+27\cdot 1^2$, so $s=-11$, and
			\[
			\Gamma=\frac{1}{9}(37-11-26)=0.
			\]
		Thus for $q=19$ and $37$, there is no $c\in \mathbb{F}_q\setminus \{0,1,\epsilon,\epsilon^{2}\}$ satisfying condition $C$ in Theorem \ref{th1e=3}.
	\end{remark}
	
\begin{example}\label{ex1}
	Let $q=p=73 = en+1$, where $e=3$ and $n=24$.
	The multiplicative group $\mathbb{F}_{q}^{*}=\langle 5\rangle$, and $D_{0}=\langle 5^{3}\rangle$. Moreover, $\epsilon=5^{24}=8$, which has order $3$ in $\mathbb{F}_q^*$ and $\epsilon^2=-1-\epsilon =-9$. 
	Elements of $\mathbb{F}_{73}$ may be represented as
	\[
	\mathbb{F}_{73}=\{0,\pm 1,\pm 2,\dots,\pm 35,\pm 36\}.
	\]
	Then one obtains
	\[
	D_0 = \{ \pm 1,\, \pm 3,\, \pm 7,\, \pm 8,\, \pm 9,\, \pm 10,\, \pm 17,\, \pm 21,\, \pm 22,\, \pm 24,\, \pm 27,\, \pm 30 \}.
	\]
	We now determine the set $\Theta$ of $c\in \mathbb{F}_q\setminus \{0,1,8,-9\}$ satisfying condition ($C$). Since $\epsilon=8$, $\epsilon^2=-9$, condition ($C$) becomes
	\[
	\Theta=\Bigl\{c\in \mathbb{F}_q\setminus \{0,1,8,-9\}\Bigr|\,\,\frac{1-c\epsilon}{1-c}=\frac{1-8c}{1-c}\in D_0,\,\,\,\frac{1-c\epsilon^2}{1-c}=\frac{1+9c}{1-c}\in D_0\Bigr\}
	\]
	and directly calculating, we have
	\begin{equation*}\label{exam1}
	\Theta=\Bigl\{-2,-4,-16,-32,18,36\Bigr\},
	\end{equation*}
	so $|\Theta|=6.$ On the one hand, from Theorem \ref{cmany}, since  $4q=292=7^2+27\cdot3^2$, we have $s=7$, and hence $$|\Theta|=\frac{1}{9}(73+7-26)=6,$$ which is also coincident with the calculation by Magma. Moreover, from Theorem \ref{th1e=3}, for each $c\in\Theta$, the power functions $f(x)=x^{25}$ and $f(x)=x^{49}$ over $\mathbb{F}_{73}$ satisfy $\Delta(f,c)\le 3$, which is confirmed by our Magma computations.
	
\end{example}

Now we consider case $c=-1$. The condition ($C$) in Theorem \ref{th1e=3} becomes that 
\[
1-c=2,\,\,\,1-c\epsilon=1+\epsilon=-\epsilon^2,\,\,\,\text{ and } 1-c\epsilon^2=1+\epsilon^2=-\epsilon
\]
belong to a common coset for some $0\leq i\leq 2$. Since $6\mid n$, we have $-1\in D_0$, and since $3\mid n$, we have $\epsilon^2,\epsilon\in D_0$. Thus  $-\epsilon^2,-\epsilon\in D_0$, i.e., $i=0$. Hence, the condition ($C$) is reduced to $2\in D_0$. By Lemma \ref{clc} (3), this is equivalent to the representation $q = s^2+27t^2$. 
 In summary, we obtain the following result.
 \begin{theorem}\label{th1e=3c=-1}
 	Let $q=p^m=3n+1$, where $p\geq3$ an odd prime and $6\mid n$. Suppose that $q = s^2+27t^2$ for some integers $s, t$. 
 	Then for $c=-1$, $d=n+1$ or $d=2n+1$, the power function $f(x)=x^d$ over $\mathbb{F}_q$ has $(-1)$-differential uniformity $\Delta(f,-1)\leq 3$.
 \end{theorem}
 \begin{remark}
 	The least two examples of such $q$ are $109$ and $127$. Indeed,
 	\[
 	109=3\cdot36+1=1^2+27\cdot2^2,\,\,\,\,127=3\cdot42+1=10^2+27\cdot1^2.
 	\]
 \end{remark}
 
 At the end of this section, we discuss the case $p=2$ and $e=3$. In this case, $\mathbb{F}_4=\{0,1,\epsilon,\epsilon^2\}$ where $\epsilon,\epsilon^2$ are elements of order $3$ satisfying $1+\epsilon+\epsilon^2=0$. We consider extensions $\mathbb{F}_q$ of $\mathbb{F}_4$ with $q=4^m=2^{2m}$ for $m\geq2$, and take $c\in \mathbb{F}_q\setminus\mathbb{F}_4$. Now conditions (1) and (2) of Theorem \ref{generalcons} become:
 
    $(1)$ $1+c\epsilon^\lambda$ ($0\leq\lambda\leq 2$) belong to a common coset $D_i$ for some $0\leq i\leq 2$;
    
    $(2)$ $1+\epsilon=\epsilon^2$ and $1+\epsilon^2=\epsilon$ belong to a common coset $D_j$ for some $0\leq j\leq 2$,\\
    where $\mathbb{F}_q^*=\langle \pi\rangle$, $D_0=\langle \pi^3\rangle$, $D_\lambda=\pi^\lambda D_0$ for $0\leq\lambda\leq 2$.
    
    Condition (2) implies that $\frac{1+\epsilon}{1+\epsilon^2}=\epsilon\in D_{j-j}=D_0$. Recall our notation $q-1=3n$ and $\epsilon=\pi^n$, combined with the fact that $\epsilon$ has order $3$, we deduce $3\mid n$, so $q=3n+1\equiv 1 \pmod{9}$. For $q=4^m=2^{2m}$ with $m\geq2$, the divisibility $9\mid q-1$ holds if and only if $3\mid m$, which gives $\mathbb{F}_{64}\subseteq\mathbb{F}_q$. 
    Thus we have the following result.
	\begin{theorem}\label{th1e=3p=2}
		Let $q=2^{6m}$, where $m\geq1$, so that $q=3n+1$ with $3\mid n$. Let $\mathbb{F}_q^*=\langle \pi\rangle$, $D_0=\langle \pi^3\rangle$, and let $\epsilon=\pi^n$ is an element of order $3$ in $\mathbb{F}_q^*$. For $c\in \mathbb{F}_q\setminus\mathbb{F}_4$, if 
		\[
		\frac{1+c\epsilon}{1+c}\in D_0,\,\,\,\,\,\,\frac{1+c\epsilon^2}{1+c}\in D_0,
		\]
	 then for $d=n+1$ and $d=2n+1$, the power function $f(x)=x^d$ over $\mathbb{F}_q$ has $c$-differential uniformity $\Delta(f,c)\leq 3$. Moreover, the number of such $c\in \mathbb{F}_q\setminus\mathbb{F}_4$ is 
	 \[
	 \Gamma = \frac19\big(q + s - 26\big),
	 \quad \text{where} 
	 \begin{cases}
	 	s = 2\sqrt{q}= 2^{3m+1}, \,\,& 2\nmid m,\\
	 	s = -2\sqrt{q}=-2^{3m+1},\,\, & 2\mid m.
	 \end{cases}
	 \]
	\end{theorem}
		{\bf Proof.} The proof of the counting formula $\Gamma = \frac19(q+s-26)$ from Theorem \ref{cmany} carries over to the case $p=2$. It remains to determine the value of $s$. Since 
		\begin{align*}
			4q = 4\cdot 2^{6m}=2^{6m+2}=(2^{3m+1})^2,
		\end{align*}
		and $s$ satisfies $4q = s^2 + 27t^2$ with $s\equiv 1\pmod{3}$, we must have $t=0$ (because $4q$ is a power of $2$ and $27t^2$ is either $0$ or has an odd factor). Hence $s=\pm2\sqrt{q}=\pm2^{3m+1}$. The sign is determined by the congruence $s\equiv 1\pmod{3}$. Since
		\[
		2^{3m+1}\equiv (-1)^{m+1}\pmod{3},
		\]
		we obtain
		\[
		s=
		\begin{cases}
			2^{3m+1}, & 2\nmid m,\\
			-2^{3m+1}, & 2\mid m.
		\end{cases}
		\]
		This completes the proof.
		
			$\hfill\Box$
			
		\begin{example}\label{ex2}
			Let $q=2^6=64 = en+1$, where $e=3$ and $n=21$. 
			We choose $g(x)=x^6+x^5+1$ a primitive polynomial of degree $6$ over $\mathbb{F}_{2}$, to construct $\mathbb{F}_{64}$. Let $\pi$ be a root of $g(x)$ in $\mathbb{F}_{64}$, then $\mathbb{F}_{q}^{*}=\langle \pi\rangle$. Let $D_{0}=\langle \pi^{3}\rangle$, Take $\epsilon=\pi^n=\pi^{21}$ , which
			is an element of order $3$ in $\mathbb{F}_{64}$.
			We now determine the set $\Sigma$ of $c\in \mathbb{F}_{64}\setminus \mathbb{F}_{4}$ which the hypotheses of which the hypotheses of Theorem \ref{th1e=3p=2} hold: 
			\[
			\Sigma=\Bigl\{c\in \mathbb{F}_{64}\setminus \mathbb{F}_{4}\Bigr|\,\,\frac{1+c\epsilon}{1+c}\in D_0,\,\,\,\frac{1+c\epsilon^2}{1+c}\in D_0\Bigr\}.
			\]
			From Theorem \ref{th1e=3p=2}, since  $4q=256=s^2+27\cdot t^2$ with $s\equiv 1\pmod{3}$, we have $s=16$, and hence $$|\Sigma|=\frac{1}{9}(64+16-26)=6.$$  which is coincident with the calculation by Magma. 
			
			To simplify the computation, we consider the action of a group on $\Sigma$.
			The maps
			\[
			\eta(x)=x^2,\qquad \delta(x)=x^{-1}
			\]
			are permutations of $\mathbb F_{64}^*=\{1,\pi,\pi^2,\dots,\pi^{62}\}$ (with $\pi^{63}=1$).
			They generate a group 
			\[
			G=\langle\eta,\delta\rangle=\{\eta^i\delta^j\,|\,0\leq i\leq5,\,j=0,1\}
			\]
			of order $12$.
			Since $\eta(\epsilon)=\epsilon^2$,  $\delta(\epsilon)=\epsilon^{-1}=\epsilon^2$, and $\eta(1)=\delta(1)=1$, 
			both $\eta$ and $\delta$ preserve $\mathbb F_4^*=\{1,\varepsilon,\varepsilon^2\}$,
			and hence $G$ acts on $\mathbb F_{64}\setminus\mathbb F_4$.
			
			We now observe that, for any $c\in\mathbb F_{64}\setminus\mathbb F_4$,
			\begin{itemize}
				\item[(I)] $c\in\Sigma$ if and only if $\eta(c)=c^2\in\Sigma$;
				\item[(II)] $c\in\Sigma$ if and only if $\delta(c)=c^{-1}\in\Sigma$.
			\end{itemize}
			Indeed, assertion (I) follows by squaring:
			\[
			\left(\frac{1+c\epsilon}{1+c}\right)^2=\frac{1+c^2\epsilon^2}{1+c^2},\quad
			\left(\frac{1+c\epsilon^2}{1+c}\right)^2=\frac{1+c^2\epsilon}{1+c^2}.
			\]
			Since $D_0$ is closed under squaring, we have
			\[
			\frac{1+c^2\epsilon^2}{1+c^2}\in D_0 \iff \frac{1+c\epsilon}{1+c}\in D_0,\quad
			\frac{1+c^2\epsilon}{1+c^2}\in D_0 \iff \frac{1+c\epsilon^2}{1+c}\in D_0,
			\]
			which implies $c^2 \in \Sigma \iff c\in \Sigma$.
			
			For (II), the identities
			\[
			\frac{1+c^{-1}\epsilon}{1+c^{-1}} = \epsilon\cdot \frac{1+c\epsilon^2}{1+c},\quad
			\frac{1+c^{-1}\epsilon^2}{1+c^{-1}} = \epsilon^2\cdot \frac{1+c\epsilon}{1+c}
			\]
			yield the equivalences
			\[
			\frac{1+c^{-1}\epsilon}{1+c^{-1}}\in D_0 \iff \frac{1+c\epsilon^2}{1+c}\in D_0,\quad
			\frac{1+c^{-1}\epsilon^2}{1+c^{-1}}\in D_0 \iff \frac{1+c\epsilon}{1+c}\in D_0,
			\]
			because $\epsilon,\epsilon^2 \in D_0$.
			Therefore $c^{-1} \in \Sigma \iff c\in \Sigma$.
			This means that  $G$ acts on $\Sigma$ as a permutation group, 
			so $\Sigma$ is a disjoint union of $G$-orbits.
			A direct computation shows that $G$ has exactly two orbits of size $6$ on $\mathbb{F}_{64}\setminus\mathbb{F}_4$, denoted $O_1$ and $O_2$.
			\[
			O_1=\{\pi^7,\pi^{14},\pi^{28},\pi^{49},\pi^{35},\pi^{56}\},\qquad
			O_2=\{\pi^9,\pi^{18},\pi^{36},\pi^{45},\pi^{54},\pi^{27}\}.
			\]
			Since $|\Sigma|=6$, the set $\Sigma$ must coincide with one of the two size-$6$ orbits.
			
			To determine which orbit is $\Sigma$, it suffices to test a single representative. Take $c=\pi^7\in O_1$. Then
			\[
			1+c\varepsilon=1+\pi^7\pi^{21}=1+\pi^{28},\qquad
			1+c=1+\pi^7,\qquad
			1+c\varepsilon^2=1+\pi^7\pi^{42}=1+\pi^{49}.
			\]
			A direct computation in $\mathbb F_{64}$ yields
			\[
			\frac{1+c\varepsilon}{1+c}=\pi^6\in D_0=\langle\pi^3\rangle,
			\qquad
			\frac{1+c\varepsilon^2}{1+c}=\pi^{-33}\in D_0.
			\]
			Hence $\pi^7\in\Sigma$, and therefore $\Sigma=O_1$.
			Explicitly,
			\[
			\Sigma=\{\pi^7,\pi^{14},\pi^{28},\pi^{49},\pi^{35},\pi^{56}\}.
			\]
			
			Consequently, for each $c=\pi^i$ with $i\in\{7,14,28,49,35,56\}$,
			the power function $f(x)=x^d$ on $\mathbb F_{64}$ (with $d=22$ or $d=43$)
			satisfies $\Delta(f,c)\le 3$, which is confirmed by our Magma computations.

			
			\end{example}
			
\section{Conclusion}\label{sec5}

In this paper, we present a general result (Theorem \ref{generalcons}) stating that for each integer $e\ge 3$ and finite field $\mathbb{F}_q$ satisfying $q-1=en$ and $e\mid n$, there exists $c\in \mathbb{F}_q\setminus \{0,1,\epsilon,\dots,\epsilon^{e-1}\}$, where $\epsilon$ is an element of order $e$, such that for every $d=ln+1$ with $1\le l\le e-1$ and $\gcd(l,e)=1$, the power function $f(x)=x^d$ over $\mathbb{F}_q$ satisfies
$\Delta(f,c)\le e,$ 
provided that conditions (1) and (2) of Theorem \ref{generalcons} hold. (If $\gcd(l,e)>1$, the parameter $e$ appearing in the exponent $d$ may be replaced by a smaller divisor.)

We then show that conditions (1) and (2) are mild. Indeed, for each fixed $e\ge 3$ and prime $p\nmid e$, there exist infinitely many extension fields $\mathbb{F}_q$ of $\mathbb{F}_p$ such that these conditions hold for some $c\in \mathbb{F}_q\setminus \{0,1,\epsilon,\dots,\epsilon^{e-1}\}$. Furthermore, using the classical Weil bound for multiplicative character sums, we prove that for sufficiently large $q$, the upper bound $\Delta(f,c)\le e$ is tight (Theorem \ref{th1tight}).

In section \ref{secc=-1}, we also discuss the important special case $c=-1$. In this setting, conditions (1) and (2) simplify to condition (B) stated in Theorem \ref{th1case-1}. For the particular case $e=3$, these conditions reduce further to condition (C) in Theorem \ref{th1e=3}, and the number of elements $c$ satisfying condition (C) is explicitly determined in Theorem \ref{cmany}.

Motivated by Theorem \ref{generalnot}, we pose the following problem. For a fixed integer $e\ge 3$, are there infinitely many primes $p\equiv 1\pmod e$ such that $\mathbb{F}_p^*$ contains an element $\epsilon$ of order $e$, and there exists $c\in \mathbb{F}_p\setminus \{0,1,\epsilon,\dots,\epsilon^{e-1}\}$ satisfying conditions (1) and (2) of Theorem \ref{generalcons}? By Theorem \ref{th1e=3c=-1}, the answer is affirmative for $e=3$ and $c=-1$, assuming there exist infinitely many primes $p\equiv 1\pmod{18}$ of the form $p=s^2+27t^2$.

A further, and more challenging, problem is to determine the $c$-differential spectrum of the power functions constructed in Theorem \ref{generalcons}. For $e=2$, such spectra have been computed using elliptic curves over finite fields (see, e.g.,\cite{WangY2026DAM}). For $e\ge 3$, the problem may be related to the study of algebraic curves of the form $y^e=g(x)$ for some $g(x)\in\mathbb{F}_q[x]$.

\end{document}